\documentclass[12pt]{elsarticle}
\usepackage[margin=2.5cm]{geometry}
\usepackage{lipsum}

\makeatletter
\def\ps@pprintTitle{%
   \let\@oddhead\@empty
   \let\@evenhead\@empty
   \def\@oddfoot{\reset@font\hfil\thepage\hfil}
   \let\@evenfoot\@oddfoot
}
\makeatother

\usepackage{amsfonts,color,morefloats,pslatex}
\usepackage{amssymb,amsthm, amsmath,latexsym}
\usepackage{amsmath,amscd}
\usepackage{enumitem}

\newtheorem{theorem}{Theorem}
\newtheorem{lemma}[theorem]{Lemma}
\newtheorem{proposition}[theorem]{Proposition}
\newtheorem{corollary}[theorem]{Corollary}

\newtheorem{definition}[theorem]{Definition}
\newtheorem{example}[theorem]{Example}

\newcommand{\rank}{{\mathrm{rank}}}
\newcommand{\lcm}{{\mathrm{lcm}}}
\newcommand{\tr}{{\mathrm{Tr}}}

\newcommand{\gf}{{\mathrm{GF}}}
\newcommand{\PG}{{\mathrm{PG}}}

\newcommand{\GAut}{{\mathrm{Aut}}}

\newcommand{\GL}{{\mathrm{GL}}}

\newcommand{\wt}{{\mathtt{wt}}}

\newcommand{\m}{\mathbb{M}}

\newcommand{\cP}{{\mathcal{P}}}
\newcommand{\cB}{{\mathcal{B}}}

\newcommand{\C}{{\mathcal{C}}}

\newcommand{\bc}{{\mathbf{c}}}

\newcommand{\bzero}{{\mathbf{0}}}

\newcommand{\bD}{{\mathbb{D}}}

\newcommand{\PGL}{{\mathrm{PGL}}}

\usepackage{blindtext}

\begin{document}

\begin{frontmatter}




\title{On Infinite Families of Narrow-Sense Antiprimitive BCH Codes Admitting $3$-Transitive Automorphism Groups and their Consequences}

\tnotetext[fn1]{
C. Ding's research was supported by the Hong Kong Research Grants Council,
Proj. No. 16302121.  C. Tang's research was supported by The National Natural Science Foundation of China (Grant No.
11871058).
}
\author[qliu]{Qi Liu}
\ead{liuqijichushuxue@163.com}

\author[cding]{Cunsheng Ding}
\ead{cding@ust.hk}

\author[sihem]{Sihem Mesnager}
\ead{smesnager@univ-paris8.fr}

\author[cmt]{Chunming Tang}
\ead{tangchunmingmath@163.com}

\author[vdt]{Vladimir D. Tonchev}
\ead{tonchev@mtu.edu}

\address[qliu]{School of Mathematics and Information, China West Normal University, Nanchong, Sichuan,  637002, China}

\address[cding]{Department of Computer Science and Engineering, The Hong Kong University of Science and Technology, Clear Water Bay, Kowloon, Hong Kong, China}
\address[sihem]{Department of Mathematics, University of Paris VIII, 93526 Saint-Denis,  University
Sorbonne Paris Cit\'e, \\
Laboratory Analysis,  Geometry and Applications (LAGA),
UMR 7539, CNRS,  93430 Villetaneuse,  and Telecom Paris, Polytechnic Institute of Paris, 91120 Palaiseau,  France}

\address[cmt]{School of Mathematics and Information, China West Normal University,  Nanchong, Sichuan,  637002, China,  and also
Department of Computer Science and Engineering, The Hong Kong University of Science and Technology, Clear Water Bay,  Kowloon, Hong Kong, China}

\address[vdt]{Department of Mathematical Sciences, Michigan Technological University, Houghton, MI 49931, USA}




\begin{abstract}

The Bose-Chaudhuri-Hocquenghem (BCH) codes are a well-studied subclass of cyclic codes that have found numerous applications in error correction and notably in quantum information processing. They are widely used in data storage and communication systems.
A subclass of attractive BCH codes is the narrow-sense BCH codes over the Galois field $\gf(q)$ with length $q+1$, which are closely related to the action of the projective general linear group of degree two on the projective line. Despite its interest, not much is known about this class of BCH codes.
 This paper aims to study some of the codes within this class and specifically narrow-sense antiprimitive BCH codes (these codes are also linear complementary duals  (LCD)  codes that have interesting practical recent applications in cryptography, among other benefits).   We shall use tools and combine arguments from algebraic coding theory, combinatorial designs, and group theory (group actions, representation theory of finite groups, etc.) to investigate narrow-sense antiprimitive BCH Codes and extend results from the recent literature. Notably, the dimension, the minimum distance of some $q$-ary BCH codes with length $q+1$, and their duals are determined in this paper. The dual codes of the narrow-sense antiprimitive BCH codes derived in this paper include almost MDS codes. Furthermore, the classification of $\PGL(2, p^m)$-invariant codes over $\gf(p^h)$ is completed.
 As an application of this result, the $p$-ranks of all  incidence structures invariant under the projective general linear group $\PGL(2, p^m)$ are determined.
 Furthermore, infinite families of narrow-sense BCH codes admitting a $3$-transitive automorphism group are obtained.
 Via these BCH codes, a coding-theory approach to constructing the Witt spherical geometry designs is presented.
 The BCH codes proposed in this paper are good candidates for permutation decoding, as they have a relatively large group of automorphisms.
\end{abstract}


\end{frontmatter}


\section{Introduction}

An \emph{$[n,k]_{q}$ linear code} $\mathcal C$ is a $k$-dimensional vector subspace of $\gf(q)^{n}$, where $q$ is a prime power.
If the linear code $\C$ has minimum distance $d$, it is also called an $[n, k, d]_q$ code.
The \emph{dual code} $\mathcal C^{\perp}$ of a linear code $\mathcal C$ is the set of vectors orthogonal to all codewords of $\mathcal C$, i.e.,
\[ \mathcal C^{\perp} = \{\mathbf w \in \gf(q)^{n} :  \langle \mathbf c, \mathbf w \rangle =0 \text{ for all } \mathbf c \in \mathcal C\}, \]
where $ \langle \mathbf c, \mathbf w \rangle$ is the usual Euclidean inner product of $\mathbf c$ and $ \mathbf w $.
A cyclic code $\C$ of length $n$ over $\gf(q)$ is a linear subspace of $\gf(q)^n$ such that  $(c_0, c_1, \ldots, c_{n-1}) \in \C$ implies that $(c_{n-1}, c_0, \ldots, c_{n-2}) \in \C$. By definition, a cyclic code is a special linear code.  Cyclic codes are widely employed in communication systems, storage devices and consumer electronics, as they have efficient encoding and decoding algorithms. If we identify an $n$-tupe $\bc=(c_0,  \dots,  c_{n-1}) \in \gf(q)^n$ with the polynomial $\bc(x)=\sum_{i=0}^{n-1} c_i x^i$ in
 the residue class ring $\gf(q)[x]/\langle x^n-1 \rangle$, any cyclic code $\C$ of length $n$ over $\gf(q)$ is an ideal in
$\gf(q)[x]/\langle x^n-1 \rangle$. Since the ring $\gf(q)[x]/\langle x^n-1 \rangle$ is principal, there is an unique monic divisor
$g(x)$ of $x^n-1$ of the smallest degree such that $\C=\langle g(x) \rangle$.  This polynomial $g(x)$ is called the \emph{generator polynomial} of $\C$ and
$h(x):=(x^n-1)/g(x)$ is referred to as the \emph{check polynomial} of $\C$. It is easily seen that the reciprocal of $h(x)$
is the generator polynomial of the dual code $\C^\perp$.

Let $n$ and $q$ be coprime.
Let $\beta$ be a primitive $n$th root of unity in an extension field of $\gf(q)$.  The Bose-Chaudhuri-Hocquenghem (BCH) code $\C_{(q, n, \delta, h)}$ with designed
distance $\delta$ consists of the set of all $\bc(x) \in \gf(q)[x]/\langle x^n-1 \rangle$ such that $\bc (\beta^{h+i})=0$
 for all $i$ in the range $0\le i \le \delta-2$,  where $h$ is an arbitrary integer and $\delta$ is a positive integer with  $2 \leq \delta \leq n$.  By definition, the BCH code $\C_{(q, n, \delta, h)}$ has generator polynomial
 $$
 \lcm\{\m_{\beta^h}(x), \m_{\beta^{h+1}}(x), \ldots,   \m_{\beta^{h+\delta-2}}(x)\},
 $$
 where $\m_{\beta^j}(x)$ denotes the minimal polynomial of $\beta^j$ over $\gf(q)$, and $\lcm$ denotes the least common
 multiple of a set of polynomials.

 It follows from the BCH bound on cyclic codes and the definition of BCH codes that $\delta$ is a lower bound on the
 minimum distance of $\C_{(q, n, \delta, h)}$ \cite{Pet61}.  When $h = 1$,  the corresponding BCH code $\C_{(q, n, \delta, h)}$ is said to be narrow-sense.
 If $n=q^m- 1$,  then $\C_{(q,n, \delta, h)}$ is referred to as a primitive
BCH code.  If $n=q^m+1$,  then $\C_{(q,n, \delta, h)}$ is called an antiprimitive BCH code.

  The discovery of BCH codes by Bose and Ray-Chaudhuri \cite{BC60} and independently Hocquenghem \cite{H59}
has been an outstanding success in the construction of codes based on algebraic structures. An attractive feature of BCH codes is that one can infer valuable information on their minimum distances and dimensions from their design parameters $q$, $n$, $\delta$, and $h$. However,
determining the actual minimum distance of most BCH codes is a challenging problem (see \citep{Charpin98}).
For BCH codes $\C_{(q, n, \delta, h)}$, only the lower bound $\delta$ on their minimum distances is known, and the actual minimum distance is known only in special cases \cite{Li17,Noguchi21,Yan18}. In particular, we have very limited
knowledge about narrow-sense antiprimitive BCH codes \cite{LL19}.

In this paper, we will consider the narrow-sense antiprimitive $q$-ary BCH codes $\C_{(q, q+1, \delta, 1)}$ of length $q+1$
with designed distance $\delta$. Some of these codes have a relatively large automorphism group. Such codes are of interest from various points of view and have certain advantages since the number of computations needed for encoding and decoding can be considerably reduced when the automorphism group is sufficiently large \cite{Macw64}.
Employing the action and representation of the finite linear group of degree two,
we will show that the supports of the codewords of minimum weight in the BCH code $\C_{(\delta^m, \delta^m+1, \delta, 1)}$
hold the Witt spherical geometry design $S(3, \delta+1,  \delta^m+1)$, which answers the question whether
there is a coding-theory construction for the Witt spherical geometry designs in the affirmative.  A short description of the spherical geometry designs given by Wiitt can be found in \cite{Witt}.
We will also present a complete classification of $\PGL(2, p^m)$-invariant $p^h$-ary codes and
derive the $p$-ranks of $\PGL(2, p^m)$-invariant combinatorial designs.
The codes treated in this paper are linear complementary dual (LCD) codes, which are important in coding theory for theoretical \cite{CMYQ18} and practical reasons \cite{CG14} (especially, as discovered in cryptography, against side-channel attacks and fault injection attacks). This paper also generalizes and extends the results in \cite{DTT21}.

This article  is organized as follows.  Sec. \ref{sec:prel} introduces fundamental notions of algebraic coding theory, combinatorial designs, and
group actions. Sec. \ref{sec:action} considers actions and representations of the finite linear group of degree two.
Sec. \ref{sec:invariant} completes the classification of $\PGL(2, p^m)$-invariant codes over $\gf(p^h)$
and gives the $p$-rank of $\PGL(2, p^m)$-invariant $t$-designs.
Sec. \ref{sec:BCH} investigates the parameters and automorphisms of the BCH codes studied in this paper
and presents a coding-theory construction of the Witt spherical geometry designs. Finally, Sec. \ref{sec:concl} concludes this paper and explains an important motivation of constructing a linear code supporting a known $t$-design.

\section{Preliminaries}\label{sec:prel}
Throughout this paper,  $p$ is a prime and $\gf (q)$  is the finite field of order $q$, where $q=p^m$ for some positive integer $m$. The set of non-zero elements of $\gf (q)$ is denoted by $\gf (q)^*$. The main goal of this paper is to push further the investigation about the narrow-sense BCH codes for which the information on them is still thin. To achieve the objective of this paper, we need to introduce basic notions of algebraic coding theory, combinatorial designs, and
group actions in this section. For additional background on these subjects, the reader is referred to \cite{CRC,Dingbook18,HP03,Pet61}.

\subsection{Linear codes and combinatorial $t$-designs}
An $[n, k, d]$ linear code $\mathcal C$ over $\gf (q)$ is a linear subspace of
$\gf (q)^n$ with dimension $k$ and minimum (Hamming) distance $d$. An $[n,k,n-k+1]$ linear code is called a maximum distance separable (MDS) code. An $[n,k,n-k]$ linear code is said to be almost maximum distance separable (almost MDS, for short).
Given a linear code $\mathcal C$ of length $n$ over $\gf (q)$, its (Euclidean) dual code  is denoted by $\mathcal C^ {\perp}$. The code $\mathcal C^ {\perp}$  is  defined by
\begin{align*}
\mathcal C^ {\perp}=\{ (b_0, b_1, \ldots, b_{n-1})\in \gf (q)^n: & \langle \mathbf c, \mathbf b \rangle:= \sum _{i=0}^{n-1} c_i  b_i=0, &\forall (c_0, c_1, \ldots, c_{n-1}) \in \mathcal C \}.
\end{align*}
A linear  complementary dual  code (abbreviated LCD) is defined as a linear code $\mathcal C$ whose dual code $\mathcal C ^ \perp$ satisfies $\mathcal C \cap \mathcal C^ \perp=\{\mathbf{0}\}$.
Let $v$ be a positive integer, $\mathbf{a}=(a_0, \ldots, a_{v-1}) \in \left ( \gf (q)^* \right )^v$ and $\C$ be a $[v,k]_{q}$ linear code.
Let $\mathbf{a} \cdot \mathcal C$ denote the linear code
$\left \{ (a_0 c_0, \ldots, a_{v-1}c_{v-1}): ( c_0, \ldots, c_{v-1}) \in \mathcal C \right \}.$
It is a simple matter to check that
\begin{eqnarray}\label{eq:a-code-dual}
\left ( \mathbf{a} \cdot \mathcal C \right )^{\perp} =  \mathbf{a}^{-1} \cdot  \mathcal C ^{\perp},
\end{eqnarray}
where $\mathbf{a}^{-1}=(a_0^{-1}, \ldots, a_{v-1}^{-1})$.
Let $\C$ be a $[v,  k,  d]$ linear code over $\gf(q)$.  Let $\cP$ be  the set of coordinate positions of codewords of $\C$
and let $\mathrm{Sym}(\cP)$ be the symmetric group acting on $\cP$.  An
element $\bc$ of $\C$ could be written as $\bc=(c_x)_{x\in \cP}$.  The permutation group $ \mathrm{PAut}(\C)$ of $\C$ is the subgroup of
$\mathrm{Sym}(\cP)$ which leaves the code globally invariant.  More precisely,  it is the
subgroup of those $g$ satisfying
\begin{eqnarray*}
g(c_x)_{x\in \cP} = (c_{g^{-1}x})_{x \in \cP} \in \C  \text{ for all }  (c_x)_{x\in \cP} \in \C.
\end{eqnarray*}
The monomial automorphism group $\mathrm{MAut}(\C)$ of $\C$ is the subgroup of $(\gf(q)^*)^n \rtimes  \mathrm{Sym}(\cP)$ which leaves
the code globally invariant.  More precisely,  it is the
subgroup of those $\left((a_x)_{x\in \cP};  g\right)$ satisfying
\begin{eqnarray*}
\left((a_x)_{x\in \cP};  g\right) (c_x)_{x\in \cP}  = (a_{x} c_{g^{-1}x})_{x \in \cP} \in \C  \text{ for all }  (c_x)_{x\in \cP} \in \C.
\end{eqnarray*}
Let $\mathrm{Gal}(\gf(q))$ denote the Galois group of $\gf(q)$ over its prime field.  The automorphism group $\mathrm{Aut}(\C)$ of $\C$ is the subgroup of $(\gf(q)^*)^n \rtimes \left(  \mathrm{Sym}(\cP) \times \mathrm{Gal}(\gf(q))\right)$ which maps $\C$ onto itself,  More precisely,  it is the
subgroup of those $\left((a_x)_{x\in \cP};  g,  \gamma \right)$ satisfying
\begin{eqnarray*}
\left((a_x)_{x\in \cP};  g,  \gamma \right) (c_x)_{x\in \cP}  = (a_{x} \gamma(c_{g^{-1}x}))_{x \in \cP} \in \C  \text{ for all }  (c_x)_{x\in \cP} \in \C.
\end{eqnarray*}
We say that $\GAut(\C)$ is  \textit{$t$-homogeneous\index{$t$-homogeneous}}
(respectively, \textit{$t$-transitive\index{$t$-transitve}}) if for every
pair of $t$-element sets of coordinates (respectively, $t$-element ordered sets of coordinates),  there is an element $\left((a_x)_{x\in \cP};  g, \gamma\right)$ of
the automorphism group
$\GAut(\C)$ such that its permutation part $g$ sends the first set to the second set.

Let $\cP$ be a set of $v \ge 1$ elements, and let $\cB$ be a set of $k$-subsets of $\cP$, where $k$ is
a positive integer with $1 \leq k \leq v$. Let $t$ be a positive integer with $t \leq k$. The pair
$\bD = (\cP, \cB)$ is called a $t$-$(v, k, \lambda)$ {\em design\index{design}}, or simply {\em $t$-design\index{$t$-design}}, if every $t$-subset of $\cP$ is contained in exactly $\lambda$ elements of
$\cB$. The elements of $\cP$ are called points, and those of $\cB$ are referred to as blocks.
A $t$-$(v,k,\lambda)$ design is referred to as a {\em Steiner system\index{Steiner system}} if $t \geq 2$ and $\lambda=1$, and is denoted by $S(t,k, v)$.

The interplay between coding theory and $t$-designs started many years ago.  Let $\C$ be a $[v, \kappa, d]$ linear code over $\gf(q)$.  Let $A_i:=A_i(\C)$ denote the
number of codewords with Hamming weight $i$ in $\C$,  where $0 \leq i \leq v$.   For each $k$ with $A_k \neq 0$,  let $\cB_k(\C)$ denote
the set of the supports of all codewords with Hamming weight $k$ in $\C$,  where the coordinates of codewords
are indexed by $(p_1, \ldots, p_v)$.  Let $\cP=\{p_1, \ldots, p_v\}$.  The pair $(\cP, \cB_k(\C))$
may be a $t$-$(v, k, \lambda)$ design for some positive integer $\lambda$,  which is called a
design supported by the code,  or shortly the support design of the code,  and is denoted by $\bD_k(\C)$.  In such a case, we say that the code $\C$ holds or supports a $t$-$(v, k, \lambda)$
design.  If $\C$ has a $t$-homogeneous
or $t$-transitive automorphism group,  the codewords of any weight $i \geq t$ of $\C$ hold a $t$-design \cite[Theorem 4.30]{Dingbook18}.

The incidence matrix $A=(a_{i,j})$ of a $t$-design $\bD$ is a $(0,1)$-matrix with rows
indexed by the blocks,  and columns indexed by the points of $\bD$,  where $a_{i,j}=1$
if the $j$th point belongs to the $i$th block, and $a_{i,j}=0$ otherwise.
If $p$ is a prime,  the $p$-rank of $\bD$ (or $\rank_{p}\bD$) is defined as
the rank of its incidence matrix $A$ over  $\gf(p)$:
 $\rank_{p}\bD = \rank_{p}A$.
Equivalently,  the $p$-rank of a design is the dimension of the linear $p$-ary code
spanned by the rows of its $(0,1)$-incidence matrix.  The $p$-rank of incidence structures, i.e., the dimension of the corresponding codes,
can be used to classify incidence structures of certain types.
For example, the $2$-rank and $3$-rank of Steiner triple and quadruple
 systems were intensively studied and
employed for counting and classifying Steiner triple and quadruple systems
 \cite{JMTW},
 \cite{JT}, \cite{SXK}, \cite{Tonch01}, \cite{Tonch03},
 \cite{Z16}, \cite{ZZ12}, \cite{ZZ13}, \cite{ZZ13a}.

\subsection{Constructions of $t$-designs from group actions}
If $X$ is a set (usually,  some kind of geometric space),  the ``symmetries ``of $X$ are often captured by the action of a group $G$ on $X$.

\begin{definition}
Given a set $X$,  and a group $G$,  a left
action of $G$ on $X$ (for short,  an action of $G$ on $X$) is
a mapping $\varphi: G \times X \rightarrow X$,  such that
\begin{enumerate}[label=(\arabic*)]
\item For all $g,h \in G$ and all $x\in X$,
\begin{eqnarray*}
\varphi(g,  \varphi(h,x))=\varphi(gh,x);
\end{eqnarray*}
\item For all $x \in X$,
\begin{eqnarray*}
\varphi(1,x)=x,
\end{eqnarray*}
where $1 \in  G$ is the identity element of $G$.
\end{enumerate}
\end{definition}

We also call this data a (left) $G$-set $X$   or say that ``$G$ acts on X” (on the left).
To alleviate the notation,  we usually write $g (x)$ or even $gx$  for $\varphi(g,x)$.
Given an action $\varphi: G \times  X \rightarrow X$,  for every $g \in  G$,  we
have a permutation $\varphi_g$ over $X$ defined by
$\varphi_g(x) =g \cdot x$,  for all $x \in X$.
Then,  the map $g \mapsto \varphi_g$ is a group homomorphism from $G$
to the symmetric group $\mathrm{Sym}(X)$ of $X$.  With a slight
abuse of notation,  this group homomorphism $G \rightarrow \mathrm{Sym}(X)$
is also denoted $\varphi$.

Recall that a finite group $G$ acting on a set $X$ is said to be $t$-transitive if for every pair of ordered $k$-tuples of
distinct points $(x_1,  \dots,  x_t)$ and $(x'_1,  \dots,  x'_t)$ there exists an element $g\in G$ such
that $gx_i=x'_i$,  $1\le i \le t$.  In particular, a transitive group is $1$-transitive.
Let $\binom{X}{k}$ be the set of subsets of $X$ consisting of $k$ elements.
A group action of $G$ on $X$ induces an action of $G$ on the set $\binom{X}{k}$ for each $1\le k \le |X|$  and given by
$(g, B) \mapsto gB:=\lbrace gx: x \in B \rbrace$.
The group $G$ is said to act $t$-homogenously on $X$
 if $G$ acts transitively on $\binom{X}{t}$.

We recall a well-known general fact (see, e.g.  \cite[Proposition 4.6]{BJL}), that  for a $t$-homogeneous group $G$ on a finite set $X$ with
$ |X|=v$  and a subset $B$ of $X$ with $|B| =k >t$, the pair
$(X, \mathrm{Orb}_{B})$ is a
$t$-$(v, k , \lambda)$ design,
where $\mathrm{Orb}_{B}$ is the orbit of $B$ under
the action of $G$ on $\binom{X}{k}$,
$\lambda=\frac{\binom{k}{t} |G| }{\binom{v}{t} | \mathrm{Stab}_{B}|}$ and
$\mathrm{Stab}_{B}$ is the stabilizer of
$B$ for this action.   For some recent works on  $t$-designs from group actions,
we refer the reader to \cite{Tang,XLW}.

\subsection{$\PGL(2,q)$ and Witt spherical geometry designs}
The \emph{projective linear group $\PGL(2,q)$ of degree two} is defined as the group
of invertible $2\times 2$ matrices with entries in $\gf(q)$,
 modulo the scalar matrices
 $\begin{bmatrix}
a & 0\\
0 &  a
\end{bmatrix}$, where $a\in \gf(q)^*$.

Here the following convention for the action of $\PGL(2,q)$ on the projective line
$\mathrm{PG}(1,q)$
is used. A matrix
$\begin{bmatrix}
a & b\\
c &  d
\end{bmatrix}
\in \PGL(2,q)$ acts on $\mathrm{PG}(1,q)$ by
\begin{eqnarray}\label{eq:action-PGL(2,q)}
\begin{array}{c}
(x_0 : x_1) \mapsto \begin{bmatrix}
a & b\\
c &  d
\end{bmatrix} (x_0 : x_1) =   (a x_0 +b x_1 : c x_0 +d x_1),
\end{array}
\end{eqnarray}
or, via the usual identification of $\gf(q) \cup \{ \infty\}$ with  $\mathrm{PG}(1,q)$, by linear fractional transformation
\begin{eqnarray}\label{eq:action-infty}
\begin{array}{c}
 x \mapsto \frac{a x +b  }{c x +d},
\end{array}
\end{eqnarray}
with the usual conventions of defining $\frac{a(-d/c)+b}{c(-d/c)+d}  =\infty $ and $\frac{a\infty+b}{c\infty+d}= a/c$.

This is an action on the left, i.e., for $g_1, g_2 \in  \PGL(2,q)$
and $x \in  \PG(1,q)$ the following holds: $g_1 (g_2(x)) = (g_1 g_2)(x)$.
The action of $\PGL(2,q)$ on $\PG(1,q)$ defined in (\ref{eq:action-infty})
is sharply $3$-transitive, i.e.,
for any distinct $a, b, c \in \gf(q) \cup \{\infty\}$ there is $g \in \PGL(2,q)$
taking $\infty$ to $a$, $0$ to $b$,  and $1$ to $c$.
 In fact, $g$ is uniquely determined and it equals
\[ g= \begin{bmatrix}
 a(b-c) & b(c-a)\\
 b-c & c-a
 \end{bmatrix}.
\]
Thus, $\PGL(2,q)$ is in one-to-one correspondence with the set of ordered triples
$(a,b,c)$ of distinct elements in $\gf(q) \cup \{\infty\}$, and
\begin{eqnarray}\label{eq:cardinality of PGL(2,q)}
\begin{array}{c}
| \PGL(2,q)| = (q+1)q(q-1).
\end{array}
\end{eqnarray}

Put $B=\PG(1,q)$ and $\mathrm{Orb}_B=\{ gB: g \in \PGL(2,q^m)\}$.  Note that $B$  is a subset of $\PG(1,q^m)$.
Define $\bD=\left( B,  \mathrm{Orb}_B \right)$.  Since $\PGL(2, q^m)$ acts 3-transitively on $\PG(1,q^m)$,  $\bD$ is a
$3$-$(q^m+1, q+1, \lambda)$ design for some $\lambda$.  Since  $\PGL(2, q^m)$ is sharply
$3$-transitive on $\PG(1,q^m)$ and $\PGL(2, q)$ is sharply $3$-transitive on $B$,
$\PGL(2,  q)$ is the setwise stabiliser of $B$.  Consequently,  $\lambda=1$ and  $\bD=(\PG(1,q^m),  \mathrm{Orb}_B )$ is a Steiner system
$S(3,q+1, q^m+1)$.  These Steiner systems were constructed by Witt \cite{Witt},  and are called \emph{ Witt spherical geometry designs}.
A coding-theoretic construction of the Witt spherical geometry design $S(3,q+1, q^m+1)$ was given in  \cite{DingTang19} for $q=3$
and in \cite{DTT21, TD20} for $q=4$.
 Whether there exists an infinite family of linear codes
holding the Witt spherical geometry design $S(3,q+1, q^m+1)$ for $q \ge 5$ being a prime power has been an open problem.
This paper will settle this open problem by presenting
an infinite family of BCH codes  holding  the Witt spherical geometry design $S(3,q+1, q^m+1)$
for any prime power $q$ and positive integer $m \geq 2$.

\subsection{The cyclicity-defining sets and trace representations of  cyclic codes}

Given a linear code $\mathcal  C$ of length $n$ and dimension $k$ over $\gf(r)$, we define
 a linear code $  \gf(r^h) \otimes \mathcal C$ over $\gf(r^h)$ by
\begin{eqnarray}
 \gf(r^h) \otimes \mathcal C=\left \{\sum_{i=1}^{k} a_i \mathbf{c}_i: (a_1, a_2, \ldots, a_k) \in \gf(r^h)^k \right \},
\end{eqnarray}
where $\left \{ \mathbf{c}_1, \mathbf{c}_2, \ldots, \mathbf{c}_k \right \}$ is a basis of $\mathcal C$ over $\gf(r)$.
This code is independent of the choice of the basis $\left \{ \mathbf{c}_1, \mathbf{c}_2, \ldots, \mathbf{c}_k \right \}$
of $\C$, is called the \emph{lifted code} of $\mathcal C$ to $\gf(r^h)$. Clearly, $\gf(r^h) \otimes \mathcal C$
and $\C$ have the same length, dimension and minimum distance, but different weight distributions.
A trivial verification shows that if $(c_0, \ldots, c_{n-1}) \in \gf(r^h) \otimes \mathcal C $,
then  $(c_0^r, \ldots, c_{n-1}^r) \in \gf(r^h) \otimes \mathcal C $.

Let $n$ be a positive integer with $\gcd(n,r)=1$.
The \emph{order} $\mathrm{ord}_{n}(r)$ of $r$ modulo $n$ is the smallest positive integer $h$ such that $r^h \equiv  1 \pmod{n}$.
Let $\mathbb Z_n$ denote the ring of residue classes of integers modulo $n$.
The $r$-cyclotomic coset of $e\in \mathbb Z_n$ is the set $[e]_{(r,n)}= \{ r^ie \bmod{n}: 0 \le i \le \mathrm{ord}_n(r)-1\}$, where
$x \bmod{n}$ denotes the unique integer $\ell$ such that $0 \leq \ell \leq n-1$ and $x \equiv \ell \pmod{n}$.
Then any two $r$-cyclotomic cosets are either equal or disjoint.
A subset $E$ of $\mathbb Z_n$ is called $r$-invariant if
the set $\{re \bmod{n}: e \in E\}$ equals $E$, that is, $E$ is the union of some $r$-cyclotomic cosets.
A subset $\widetilde{E}=\{e_1, \ldots, e_t\}$ of an $r$-invariant set $E$ is called a complete set of representatives of $r$-cyclotomic cosets of $E$
if $[e_1]_{(r,n)}, \ldots, [e_t]_{(r,n)}$  are pairwise distinct and $E= \cup_{i=1}^{t} [e_i]_{(r,n)}$.

Let $\gamma$ be a primitive $n$-th root of unity in $\gf(r^h)$, where $h=\mathrm{ord}_{n}(r)$.
It is known \cite{HP03} that any $r$-ary cyclic code of length $n$ with $\gcd(n,r)=1$ has a simple description by means of the trace function.
The trace function $\tr_{q^r/q} : \gf(q^h) \rightarrow \mathbb \gf(q)$ is defined as:
\begin{displaymath}
\tr_{q^h/q}(x):=\sum_{i=0}^{ h-1}
  x^{q^{i}}=x+x^{q}+x^{q^2}+\cdots+x^{q^{h-1}}.
  \end{displaymath}
  The trace function from  $\mathbb {F}_{q^h}$ to its prime subfield is called the \emph{absolute trace} function.

\begin{theorem}\label{thm:cyclic-trance}  \cite{HP03}
Let $\mathcal C$ be an $[n,k]_r$ cyclic code with $\gcd(n,r)=1$ and $\gamma$ be a primitive
$n$-th root of unity in $\gf(r^h)$, where $h=\mathrm{ord}_{n}(r)$.
Then there exists a unique $r$-invariant set $E \subseteq \mathbb Z_n$ such that
\begin{eqnarray*}
\mathcal C= \left \{  \left (  \sum_{i=1}^{t} \tr_{r^{h_i}/r} \left ( a_i \gamma^{e_ij} \right ) \right )_{j=0}^{n-1}: a_i \in \gf\left (r^{h_i} \right )\right \},
\end{eqnarray*}
where $\{e_1, \ldots, e_t\}$ is any complete set of representatives of $r$-cyclotomic cosets of $E$ and $h_i=| [e_i]_{(r,n)}|$.
Moreover, $k=|E| = \sum_{i=1}^t h_i$.
\end{theorem}

Theorem \ref{thm:cyclic-trance} states that there is a one-to-one correspondence between cyclic linear codes
over $\gf(r)$ with length $n$ and $r$-invariant subsets of
$\mathbb Z_n$ with respect to a fixed $n$-th root of
unity $\gamma$.
We will call the set $E$ in Theorem \ref{thm:cyclic-trance}  the \emph{cyclicity-defining set} of $\mathcal C$ with respect to $\gamma$.

The following corollary is an immediate consequence of Theorem \ref{thm:cyclic-trance}.

\begin{corollary}\label{cor:cyclic-extension} \cite{DTT21}
Let $n$ be a positive integer such that $\gcd(n,r)=1$.
 Let $\mathcal C$ be an $[n,k]_r$ cyclic code with cyclicity-defining set $E$ and $\gf(r^{\ell}) \otimes \mathcal C$ be the lifted code of $\C$ to $\gf(r^{\ell})$.
Then $\gf(r^{\ell}) \otimes \mathcal C$ is an $[n,k]_{r^{\ell}}$ cyclic code defined by the cyclicity-defining set $E$ of $\mathcal C$. In particular,
\begin{eqnarray*}
\gf(r^{h}) \otimes  \mathcal C= \left \{  \left (  \sum_{e\in E} a_e \gamma^{je} \right )_{j=0}^{n-1}: a_e \in \gf\left (r^{h} \right )\right \},
\end{eqnarray*}
where $h=\mathrm{ord}_n(r)$ and $\gamma$ is a primitive $n$-th root of unity in $\gf(r^h)$.
\end{corollary}

Hence,  the two code $\gf(r^{\ell}) \otimes \mathcal C$ and $\mathcal C$ have the same cyclicity-defining sets.
Let $n$ be a positive integer with $\gcd(n,r)=1$ and $h=\mathrm{ord}_n(r)$.
 Let $U_n$ be the cyclic multiplicative group of all $n$-th roots of unity in $\gf(r^h)$.
By polynomial interpolation, every function $f$ from $U_{n}$ to $\gf(r)$  has a unique \emph{univariate polynomial expansion} of the form
\[ f(u)= \sum_{i=0}^{n-1} a_i u^i,\]
where $a_j \in \gf(r^h)$, $u \in U_n$.  As a direct result of Theorem \ref{thm:cyclic-trance},
we have the following conclusion concerning cyclicity-defining sets of cyclic codes.

\begin{corollary}\label{cor:a_e-neq-0-E} \cite{DTT21}
Let $n$ be a positive integer with $\gcd(n,r)=1$, $h=\mathrm{ord}_n(r)$ and $\gamma$ a primitive $n$-th root of unity in $\gf(r^h)$.
Let $\mathcal C$ be an $[n,k]_r$ cyclic code with cyclicity-defining set $E$. Let $f(u)=\sum_{i=0}^{n-1} a_i u^i \in \gf(r^h) [u]$.
If $\left (f(\gamma^j) \right )_{j=0}^{n-1} \in \mathcal C$ and $a_i\neq 0$, then $i\in E$.
\end{corollary}

\section{Group actions and representations of $\GL(2,q)$ and $\PGL(2,q)$}\label{sec:action}

This section considers actions and representations of the projective general linear group $\PGL(2,q)$.
These results will play an important role in Sections \ref{sec:invariant} and \ref{sec:BCH}.

\subsection{Stabilizers of certain subsets of the projective line $\PG(1,  q^2)$}

Let $U_{q+1}$ be the subgroup of $\gf(q^2)$ consisting of elements whose norm to $\gf(q)$ is $1$.
By Hilbert Theorem 90,  we may describe the elements in $U_{q+1}$ in terms of the elements of the projective line $\PG(1, q)$ as
\begin{eqnarray*}\label{eq:Norm-PLine}
U_{q+1}=\left( \begin{array}{cc}
u_0 & 1\\
1 & u_0
\end{array} \right) \PG(1, q)=\left\lbrace \frac{u_0x+1}{x+u_0}: x \in \PG(1, q)  \right\rbrace,
\end{eqnarray*}
with $u_0 \in U_{q+1} \setminus \{\pm 1\}$ and the convention that the quotient is $u_0$ for $x = \infty$.
Recall that the setwise stabilizer of $\PG(1,q)$ under the action of $\PGL(2,  q^2)$   on $\PG(1,q^2)$
is $\PGL(2,  q)$.  Then the two setwise stabilizer groups $\mathrm{Stab}_{U_{q+1}}$ and $\mathrm{Stab}_{\PG(1,q)}= \PGL(2,q)$ are related by
\begin{eqnarray}\label{eq:stabilizers}
\mathrm{Stab}_{U_{q+1}}= \left( \begin{array}{cc}
u_0 & 1\\
1 & u_0
\end{array} \right) \PGL(2,  q)  \left( \begin{array}{cc}
u_0 & 1\\
1 & u_0
\end{array} \right)^{-1} .
\end{eqnarray}
Let $\left( \begin{array}{cc}
a & b\\
c & d
\end{array} \right) \in \PGL(2,  q)$.  A standard computation gives
\begin{eqnarray}\label{eq:conjugation}
 \left( \begin{array}{cc}
u_0 & 1\\
1 & u_0
\end{array} \right)
\left( \begin{array}{cc}
a & b\\
c & d
\end{array} \right)
 \left( \begin{array}{cc}
u_0 & 1\\
1 & u_0
\end{array} \right)^{-1}
= \frac{u_0'}{(u_0^2-1) }  \left( \begin{array}{cc}
\overline{d}^q & \overline{c}^q\\
\overline{c} & \overline{d}
\end{array} \right),
\end{eqnarray}
where $u_0'=\sqrt{-1} u_0$,  $\overline{c}= u_0'^{-1}\left(cu_0^2+(a-d)u_0 -b\right)$, $\overline{d}= u_0'^{-1}\left(du_0^2+(b-c)u_0 -a\right)$ and $\sqrt{1} \in \gf(q^2)$.
Then the following commutative diagram is obtained:
\begin{eqnarray}\label{eq:CD}
\begin{CD}
\gf(q) \cup \lbrace \infty  \rbrace @>\frac{u_0x+1}{x+u_0}>>   U_{q+1}\\
@V\frac{ax+b}{cx+d}VV     @VV \frac{\overline{d}^qu+\overline{c}^q}{\overline{c}u+\overline{d}}V\\
\gf(q) \cup \lbrace \infty  \rbrace @>\frac{u_0x+1}{x+u_0}>>   U_{q+1}
\end{CD}
\end{eqnarray}
Write $A=\left \{ \left( \begin{array}{cc}
d^q & c^q\\
c & d
\end{array} \right) : c, d \in \gf(q^2),  c^{q+1} \neq d^{q+1} \right\}$.
Combining (\ref{eq:stabilizers}) and (\ref{eq:conjugation}) yields
$\mathrm{Stab}_{U_{q+1}} \subseteq A$.  On the other hand,   it is clear that $A \subseteq \mathrm{Stab}_{U_{q+1}}$ by Hilbert Theorem 90.
We thus deduce that $\mathrm{Stab}_{U_{q+1}} = A$.  By the commutativity of the diagram in (\ref{eq:CD}),  we see that
 the action of $\mathrm{Stab}_{U_{q+1}}$ on $U_{q+1}$ is equivalent to the action of $\PGL(2,q)$ on $\PG(1,q)$.

A summary of the above discussion implies the following.

\begin{proposition}\label{prop:Stab-U}
The setwise stabilizer $\mathrm{Stab}_{U_{q+1}}$  of $U_{q+1}$ can be expressed as
\begin{eqnarray*}
\mathrm{Stab}_{U_{q+1}}=\left \{ \left( \begin{array}{cc}
d^q & c^q\\
c & d
\end{array} \right) \in \PGL(2,  q^2):  c^{q+1} \neq d^{q+1} \right\}.
\end{eqnarray*}
Furthermore,   the action of $\mathrm{Stab}_{U_{q+1}}$ on $U_{q+1}$ is equivalent to the action of $\PGL(2, q)$ on $\PG(1,  q)$.
Hence,  $\mathrm{Stab}_{U_{q+1}}$ is sharply $3$-transitive.
\end{proposition}

\subsection{A modular representation  of $\GL(2,q)$}

Representation theory of finite groups is a potent tool for constructing linear codes invariant under a given group $G$.
Basic material related to the representation theory of groups can be found in \cite{LP10}.
\begin{definition}
A representation of a group $G$ is a pair $(V,  \rho)$ where $V$ is a vector space over a field $\mathbb F$
and $\rho$ is a map from $G\times V $ to $V$ such that
\begin{enumerate}[label=(\arabic*)]
\item $\rho$ is a group action (the action is associative and $\rho(1,v)=v$ for all $v \in V$ ),  and
 \item  the map $V \rightarrow V$  defined by $v \mapsto \rho(g,v)$ is linear
for all $g \in G$.
\end{enumerate}
\end{definition}

The dimension of $V$ over $\mathbb F$ is called the degree of the representation.
A representation is called an ordinary
representation if $\mathrm{char}(\mathbb F) \not | \,  |G|$,   and it is called a modular
representation if $\mathrm{char}(\mathbb F)  | \,  |G|$.

A representation is nothing but a linear action of $G$ on $V$.
We can also think of a group representation $(V, \rho)$ of $G$ as a group homomorphism from $G$ to $\mathrm{GL}(V)$,
where $\mathrm{GL}(V)$ denotes the group of invertible linear transformations
from $V$ to itself. If $V$ has a basis $v_1, \dots, v_n$ then we can identify $\GL(V )$ with the more familiar group $\GL_n$ of invertible
$n \times n$ matrices. Informally speaking, a representation of a group $G$ is a way of writing the group elements
as square matrices of the same size, which is multiplicative and assigns $1 \in G$ the identity
matrix.

Consider the set
\begin{eqnarray}
\begin{array}{c}
\mathcal{P}(\delta,q):= \left\lbrace  \tr_{q^2/q} \left(  \sum_{i=1}^{\delta -1}  a_i u^i\right) \in  \gf(q^2)[u]/\langle u^{q+1}-1 \rangle: a_i\in \gf(q^2)  \right\rbrace.
\end{array}
\end{eqnarray}

For $A=\left( \begin{array}{cc}
a & b\\
c & d
\end{array} \right)^{-1} \in \mathrm{GL} (2,q^2)$
and $f\in \mathcal{P}(\delta,q)$ we set
\begin{eqnarray}\label{eq:action-circ}
\begin{array}{c}
(A\circ f) (u):= (cu+d)^{(q+1)(\delta-1)} f\left(\frac{au+b}{cu+d}\right).
\end{array}
\end{eqnarray}

The following result on binomial coefficients will be used in evaluating $A\circ f$.
\begin{lemma}\label{eqn:q+e-1 mod}
Let $\delta$ be a power of a prime $p$ and let $e$ be a positive integer with $e \le \delta -1 $.  Let $s$ be an integer with $e \le s \le \delta-1$.  Then
the binomial coefficient $\binom{\delta-1+e}{s} $ is divisible by $p$.
\end{lemma}
\begin{proof}
Expanding the power $(1+x)^{\delta-1+e} \in \gf(p)[x]$ by the Binomial Theorem,  we have
\begin{eqnarray*}
\begin{array}{c}
(1+x)^{\delta-1+e}=\sum_{i=0}^{\delta-1+e} \binom{\delta-1+e}{s} x^s.
\end{array}
\end{eqnarray*}
Since $\delta$ is a power of $p$,  we get
\begin{eqnarray*}
\begin{array}{rl}
(1+x)^{\delta-1+e}&\equiv (1+x)^{\delta} (1+x)^{e-1}\\
&\equiv (1+x^{\delta})(1+x)^{e-1} \pmod p.
\end{array}
\end{eqnarray*}
Then,  we have
\begin{eqnarray}
\begin{array}{c}
\sum_{i=0}^{\delta-1+e} \binom{\delta-1+e}{s} x^s\equiv (1+x^{\delta})(1+x)^{e-1} \pmod p.
\end{array}
\end{eqnarray}
By comparing the coefficients of $x^s$ ($e \le s \le \delta-1$) on both sides we see that $ \binom{\delta-1+e}{s}  \equiv 0 \pmod p$,
  which completes the proof of the lemma.
\end{proof}

The following lemma shows that for $A=\left( \begin{array}{cc}
d^q & c^q\\
c & d
\end{array} \right)^{-1} \in \mathrm{GL} (2,q^2)$ ,  $f \mapsto A\circ f$ defines a $\gf(q)$-linear transformation over $\mathcal{P}(\delta, q)$ in some cases.
\begin{lemma}\label{lem:action-linear}
Let $q=p^m$ and $\delta$ be a power of $p$.  Let $A=\left( \begin{array}{cc}
d^q & c^q\\
c & d
\end{array} \right)^{-1} \in \mathrm{GL} (2,q^2)$  and $f\in \mathcal{P}(\delta,q)$.  Then $A\circ f \in \mathcal{P}(\delta,q)$.
\end{lemma}
\begin{proof}
In order to prove this lemma it suffices to prove the following claim:
for any integer $e$ with $1 \le e \le \delta-1$ and $\beta \in \gf(q^2)$ ,
we have $A\circ g \in \mathcal{P}(\delta,q)$,  where $g=\tr_{q^2/q} \left( \beta u^e \right)$.
By the definition of the operator $'\circ'$, we have
\begin{eqnarray}\label{Eqn:A-g}
\begin{array}{l}
\left( A\circ g \right) (u) \\
= (c u+d)^{(q+1)(\delta-1)} \tr_{q^2/q}  \left( \beta \left(\frac{d^qu+c^q}{cu+d}\right)^e \right)\\
= (c^q u^{-1}+d^q)^{\delta-1} (c u+d)^{\delta-1} \tr_{q^2/q}  \left( \beta \left(\frac{d^qu+c^q}{cu+d}\right)^e \right)\\
=\tr_{q^2/q} \left(\beta  u^e \left( c u +d \right)^{\delta-1-e} \left( c^q u^{-1} +d^q \right)^{\delta-1+e}  \right).
\end{array}
\end{eqnarray}
Using the Binomial Theorem,  we get
\[ u^e \left( c u +d \right)^{\delta-1-e}=\sum_{i=0}^{\delta-1-e} \binom{\delta-1-e}{i}a_i u^{i+e} \]
 and
\[ \left( c^q u^{-1} +d^q \right)^{\delta-1+e} =\sum_{j=0}^{\delta-1+e} \binom{\delta-1 +e}{j} b_j u^{-j},\]
where $a_i=c^id^{\delta-1-e-i}$ and $b_j=c^{qj}d^{q(\delta-1+e-j)}$.
By (\ref{Eqn:A-g}) we have
\begin{eqnarray}\label{Eqn:A-g-1}
\begin{array}{l}
\left( A\circ g \right) (u) \\
=\tr_{q^2/q} \left( \beta  \sum \limits_{ \begin{array}{c}
0 \le i \le \delta-1-e \\ 0 \le j\le \delta-1+e\end{array}}\binom{\delta-1-e}{i} \binom{\delta-1+e}{j} a_i b_j u^{i-j+e} \right).
\end{array}
\end{eqnarray}
Note that $0 \le i \le \delta-1-e$ and $0 \le j\le \delta-1+e$ imply that $-(\delta-1)\le i-j+e \le \delta -1$.
Hence there are $c_t \in \gf(q^2)$ ($1 \le t \le \delta-1$) such that
\begin{eqnarray}
\begin{array}{l}
\left( A\circ g \right) (u)
=\tr_{q^2/q}\left(\beta \sum_{i=1}^{\delta-1}c_tu^i \right)+\tr_{q^2/q}(\beta W),
\end{array}
\end{eqnarray}
where $W=\sum_{t=0}^{\delta-1-e}  \binom{\delta-1-e}{t} \binom{\delta-1+e}{t+e} a_tb_{t+e}$.
By Lemma \ref{eqn:q+e-1 mod},  we deduce that $\binom{p-1+e}{t+e} \equiv 0 \pmod p$.
It follows that  $W=0$,
 which implies that $ A\circ g  \in  \mathcal{P}(\delta,  q)$.
  This completes the proof.
\end{proof}

We consider the lifting of $\mathrm{Stab}_{U_{q+1}}$ in $\GL(2,q^2)$
\begin{eqnarray}
\mathrm{\overline{Stab}}_{U_{q+1}}=
 \left\{ \left( \begin{array}{cc}
d^q & c^q\\
c & d
\end{array} \right): c, d \in \gf(q^2),  c^{q+1} \neq d^{q+1}\right\}.
\end{eqnarray}
The following lemma implies that the action $'\circ'$ defined in (\ref{eq:action-circ}) gives  a representation of $ \mathrm{\overline{Stab}}_{U_{q+1}}$ on
the linear space $\mathcal{P}(\delta, q)$. It comes  straightforwardly from Lemma \ref{lem:action-linear} and the definition of the action $'\circ'$.

\begin{lemma}\label{lem:representation}
Let $q=p^m$ and $\delta$ be a power of $p$.  Let $A_1, A_2 \in  \mathrm{\overline{Stab}}_{U_{q+1}}$ and $f_1,  f_2 \in \mathcal{P}(\delta, q)$,   and denote by $E$ the $2\times 2$ identity matrix.  Then the following
hold:
\begin{enumerate}[label=(\arabic*)]
 \item $A_1 \circ f_1 \in \mathcal{P}(\delta, q)$ ,
\item $E \circ f_1= f_1$ ,
\item $(A_1A_2) \circ f_1 = A_1 \circ  (A_2 \circ  f_1 )$,
\item  $A_1 \circ  ( af_1+bf_2)=aA_1 \circ  f_1 +b A_2\circ  f_2$ for all $a,b \in \gf(q)$.
\end{enumerate}
\end{lemma}

By (\ref{eq:stabilizers}),  we have $\mathrm{\overline{Stab}}_{U_{q+1}}= \left( \begin{array}{cc}
u_0 & 1\\
1 & u_0
\end{array} \right) \GL(2,  q)  \left( \begin{array}{cc}
u_0 & 1\\
1 & u_0
\end{array} \right)^{-1}$.  Hence,  the representation of $\mathrm{\overline{Stab}}_{U_{q+1}}$  gives rise to a representation of $\GL(2,q)$ on $\cP(\delta,  q)$.

\section{$\PGL(2, q)$-invariant codes}\label{sec:invariant}

The main objective of this section is to classify all linear codes over $\gf(p^h)$ of length $p^m+1$ that are invariant under
$\PGL(2,  p^m)$.
As an immediate application,
 we derive
the $p$-rank of the incidence matrices of $t$-$(p^m+1, k, \lambda)$ designs that are invariant under $\PGL(2,p^m)$.

Let $\mathcal C$ be a $[p^m+1, k]_{p^h}$ linear code.
We can regard $U_{p^m+1}$ as the set of the coordinate positions of $\mathcal C$ and write the codeword of $\mathcal C$ as $\left ( c_u\right )_{u\in U_{p^m+1}}$.
Then the set of coordinate positions of $\mathcal C$ could be endowed with the
action of $\mathrm{Stab}_{U_{p^m+1}}$.  According to Proposition
\ref{prop:Stab-U},
we only need to find  all linear codes over $\gf(p^h)$ of length $p^m+1$ which are invariant under $\mathrm{Stab}_{U_{p^m+1}}$.

The following lemma gives the polynomial expansion of the linear fractional transformation
$\frac{u-c^q}{-cu+1}$, where $c \in \gf(q^2)^* \setminus U_{q+1}$.
\begin{lemma}\label{lem:frac-poly}
Let  $c \in \gf(q^2)^* \setminus U_{q+1}$.
Then for any $u\in U_{q+1}$,  the following holds:
\begin{eqnarray*}
\frac{u-c^q}{-cu+1} =\sum_{i=1}^{q} c^{i-1} u^i.
\end{eqnarray*}

\end{lemma}
\begin{proof}
An easy computation shows that
\begin{eqnarray*}
\begin{array}{rl}
\sum_{i=1}^{q} c^{i-1} u^i & = \frac{1-(cu)^{q}}{1-cu} u\\
 & = \frac{u-c^qu^{q+1}}{1-cu} \\
 &= \frac{u-c^q}{-cu+1},
\end{array}
\end{eqnarray*}
which completes the proof.

\end{proof}

The following lemma expresses the coefficients of the polynomial expansion of a function $f$ over $U_{q+1}$
in terms of the sums  over $U_{q+1}$ of the product function of $f$ and  the power functions $u^j$. It comes directly using the same argument as in the case that $q$ is a power of $2$.\cite[Lemma ]{DTT21}.
\begin{lemma}\label{lem:ai=sum-f}
Let $f$ be a function from $U_{q+1}$ to $\gf(q^{2h})$ with $h\ge 1$.  Let
$\sum_{i=0}^{q} a_i u^i$
be the polynomial expansion of $f$,  where $a_i \in \gf(q^{2h})$.  Then
$a_i=\sum_{u \in U_{q+1}} f(u) u^{-i}$,  where $0 \le i \le q$.
\end{lemma}

The following lemma gives the first two terms of the polynomial
expansion for the function $\left ( g u\right )^e$ over $U_{q+1}$,  where $g=\left( \begin{array}{cc}
1 & -c^q\\
-c & 1
\end{array} \right)^{-1}$.
\begin{lemma}\label{lem:a1-neq-0}
Let $q=p^m$ with $m \ge 1$,  $p$ a prime and $e$  an integer such that $1 \le e \le q$.  Let $g=\left( \begin{array}{cc}
1 & -c^q\\
-c & 1
\end{array} \right)^{-1}$,  where $c \in \gf(q^2)^* \setminus U_{q+1}$.
Let $a_0 + a_1 u + \cdots + a_q u^q$ be the polynomial expansion of
the function from $U_{q+1}$ to $\gf(q^2)$ given by
$u \mapsto \left (g u \right )^e$.  Then $a_0=0$ and $a_1 = c^{q(e-1)}$.

\end{lemma}

\begin{proof}
Applying Lemma \ref{lem:ai=sum-f} to the function $f(u)=\left ( g u \right )^e$, we obtain
\begin{eqnarray*}
\begin{array}{rll}
a_1 &= \sum_{u \in U_{q+1}} \left ( g u \right )^e u^{-1} &\\
&= \sum_{u \in U_{q+1}} u^e \left( g^{-1} u \right)^{-1}& \text{ Substituting } u \text{ with } g^{-1} u\\
&= \sum_{u \in U_{q+1}} u^e \frac{-cu+1}{u-c^q} & \\
&=\frac{-c}{-c^q} \sum_{u \in U_{q+1}} u^e \frac{u-1/c}{-u/c^q+1} & \\
&=\frac{c}{c^q} \sum_{u \in U_{q+1}} u^e \sum_{j=1}^q u^j/c^{q(j-1)} & \\
&=\frac{c}{c^q}\sum_{j=1}^q   \sum_{u \in U_{q+1}} u^e u^j/c^{q(j-1)} & \\
&= c^{q(e-1)}, &
\end{array}
\end{eqnarray*}
where the last equality follows from Lemmas \ref{lem:frac-poly} and \ref{lem:ai=sum-f}.

Employing Lemma \ref{lem:ai=sum-f} on $\left ( g u \right )^e$ again, we have
\begin{eqnarray*}
\begin{array}{rll}
a_0 &= \sum_{u \in U_{q+1}} \left ( g u \right )^e  &\\
&= \sum_{u \in U_{q+1}} u^e  & \text{ Substituting } u \text{ with } g^{-1} u\\
&= 0.&
\end{array}
\end{eqnarray*}
This completes the proof.

\end{proof}

Now we are ready to prove the main result of this section.

\begin{theorem}\label{thm:code-PGL-4}
Let $q=p^m$ with $m\ge 1$ and $p$ being a prime.  If $\mathcal C$ is a linear code over $\gf(p^h)$ of length $q+1$
that is invariant under the permutation action of $\PGL(2,  q)$,
then $\mathcal C$ must be one of the following:
\begin{enumerate}
\item[\textnormal{(I)}] the zero code $\mathcal C_0 = \{(0,0, \ldots  ,0)\}$; or
\item[\textnormal{(II)}] the whole space $\gf(p^h)^{q+1}$,  which is the dual of $\mathcal C_0$; or
\item[\textnormal{(III)}] the repetition code $\mathcal C_1 = \{(c,c, \ldots  ,c): c \in \gf(p^h)\}$ of dimension $1$; or
\item[\textnormal{(IV)}]  the code $\mathcal C_1^{\perp}$,  given by
\[\mathcal C_1^{\perp}= \left \{ (c_0, \ldots, c_{q})\in \gf(p^h)^{q+1}: c_0+ \cdots +c_{q}=0 \right \}.\]
\end{enumerate}

\end{theorem}
\begin{proof}
It is evident that the four trivial $p^h$-ary linear codes $\mathcal C_0, \mathcal C_0^{\perp}, \mathcal C_1$ and
$\mathcal C_1^{\perp}$ of length $q+1$ are invariant under $\PGL(2,  q)$.

Let $\mathcal C$
be a $p^h$-ary linear code of
 length $q+1$ which is invariant
under $\PGL(2,q)$,  which amounts to saying that $\mathcal C$ is invariant under $\mathrm{Stab}_{U_{q+1}}$ by Proposition \ref{prop:Stab-U}.
 Observe that  the translation $\pi(u)=u_0 u$ belongs to $\mathrm{Stab}_{U_{q+1}}$,  where $u_0 \in U_{q+1}$.
 This clearly forces $\mathcal C$ to be a cyclic code.  Let $E$ be the cyclicity-defining set of $\mathcal C$.  We consider the following four cases for $E$.

 If $E= \emptyset $, then $\mathcal C= \mathcal C_0$

 If $E= \{ 0 \}$, then $\mathcal C= \mathcal C_1$

 If $\{0\} \subsetneq E$, then there exists an $e\in E \setminus \{0\}$. Applying Corollary \ref{cor:cyclic-extension},  the lifted  code  $\gf(q^{2h}) \otimes \mathcal C$
 to $\gf(q^{2h})$ is the cyclic code over $\gf(q^{2h})$ with respect to the cyclicity-defining set $E$.
 We see at once that $\gf(q^{2h}) \otimes \mathcal C$ also stays invariant under  $\mathrm{Stab}_{U_{q+1}}$ from the definition of  lifting of a cyclic code.
Combining Corollary \ref{cor:cyclic-extension} with Proposition \ref{prop:Stab-U} we obtain
 $ \left ( (g u  )^e \right )_{ u\in U_{q+1}} \in \gf(q^{2h}) \otimes \mathcal C$, where
  $g=\left( \begin{array}{cc}
1 & -c^q\\
-c & 1
\end{array} \right)^{-1}$ and $c \in \gf(q^2)^* \setminus U_{q+1}$.
Applying Corollary \ref{cor:a_e-neq-0-E} and Lemma \ref{lem:a1-neq-0}  we can assert that $1\in E$. Thus
$ \left ( g u \right )_{ u\in U_{q+1}} \in \gf(q^{2h}) \otimes \mathcal C$.
Combining Corollary \ref{cor:a_e-neq-0-E} and Lemma \ref{lem:frac-poly} we deduce $E=\{0,1, \cdots, q\}$.
 We thus get $\mathcal C= \gf(p^h)^{q+1}$.

If $E \neq \emptyset$ and $ 0 \not \in E$, then there exists an $e\in E \setminus \{0\}$.
An analysis similar to that in the proof of the case of $\{0\} \subsetneq E$ shows
that $E=\{1, \ldots, q\}$ and $\mathcal C= \mathcal C_1^{\perp}$.
 This completes the proof.
\end{proof}

For any set $A$ and a positive integer $k$, recall that $\binom{A}{k}$ denotes the set of all $k$-subsets of $A$.
The following result is an important  consequence  of Theorem
\ref{thm:code-PGL-4}.

\begin{theorem}\label{thm;2-rank}
Let $p$ be a prime.  Let $\cB \subseteq \binom{\PG(1,  p^m)}{k}$ such that $m\ge 1$, $1\le k \le p^m$ and $\cB$ is invariant under the action of
$\PGL(2,  p^m)$.
Then the incidence structure $\bD=(\PG(1,p^m), \cB)$  has $p$-rank  $p^m$ or $p^m+1$
depending on whether $k$ is a multiple of  $p$ or not.
\end{theorem}

\begin{proof}
Since $\cB$ is invariant under the action of $\PGL(2,  p^m)$, then so is the code $\mathcal C_p(\bD)$ of $\bD$.  It then follows
from Theorem \ref{thm:code-PGL-4} that $\mathcal C_p(\bD)=\mathcal C_1^{\perp}$ or $\mathcal C_2(\bD)=\gf(p)^{p^m+1}$.
The desired conclusion then follows.
\end{proof}

\section{Narrow-sense $q$-ary BCH codes of length $q+1$ }\label{sec:BCH}
 In this section,  we shall determine parameters and automorphisms of some narrow-sense $q$-ary BCH codes of length $q+1$
 and present a coding-theory construction of the Witt spherical geometry designs.

 For a positive integer $\ell \leq q+1$, define a $2 \delta \times  \ell$ matrix $M_{\delta,  \ell}$ by
\begin{eqnarray}\label{eq:M}
\left [
\begin{array}{cccc}
u_1^{-(\delta-1)} & u_2^{-(\delta-1)} & \cdots & u_{\ell}^{-(\delta-1)} \\
\vdots & \vdots & \cdots & \vdots\\
u_1^{-2} & u_2^{-2} & \cdots & u_{\ell}^{-2} \\
u_1^{-1} & u_2^{-1} & \cdots & u_{\ell}^{-1} \\
u_1^{+1} & u_2^{+1} & \cdots & u_{\ell}^{+1} \\
u_1^{+2} & u_2^{+2} & \cdots & u_{\ell}^{+2} \\
\vdots & \vdots & \cdots & \vdots\\
u_1^{+(\delta-1)} & u_2^{+(\delta+1)} & \cdots & u_{\ell}^{+(\delta+1)} \\
\end{array}
\right ],
\end{eqnarray}
where $u_1, \ldots,  u_{\ell} \in U_{q+1}$. For all integers $r_1$ and $ r_2 $ with $-(\delta-1) \le r_1 <r_2 \le (\delta-1)$,  let $M_{\delta,  \ell}[r_1, r_2]$
denote the submatrix of $M_{\delta, \ell}$  obtained by choosing   the row
$(u_1^{i},  u_2^{i},  \ldots ,u_{\ell}^{i})$,   where $r_1 \le i \le r_2$. By adopting similar arguments as in \cite[Lemma 29]{TD20}, we deduce the following result.

\begin{lemma}\label{lem:sol-rank}
Let $M_{\delta, \ell}$ be the matrix given by (\ref{eq:M}) with  $\{u_1, \ldots,  u_{\ell}\} \in  \binom{U_{q+1}}{\ell}$.
Consider the system of homogeneous linear equations defined by
\begin{align}\label{eq:Mx=0}
M_{\delta,  \ell} (x_1, \ldots,  x_{\ell})^{T}=0.
\end{align}
Then (\ref{eq:Mx=0}) has a nonzero  solution $(x_1, \ldots,  x_{\ell})$ in $\gf(q)^{\ell}$  if and only if
$\rank (M_{\delta,  \ell})<\ell$, where $\rank (M_{\delta, \ell})$ denotes the rank of  the matrix $M_{\delta, \ell}$.
\end{lemma}

The $\ell$'th elementary symmetric polynomial $\sigma_{\ell}$ over a set $\{u_1, \dots, u_n \}$ is a specific  sum
of products without permutation of repetitions and defined by
\[\sigma_{\ell}(u_1,\dots, u_n)=\sum_{1\le i_1 < \dots <i_{\ell} \le n} u_{i_1} \cdots u_{i_{\ell}}.\]
The following lemma shows a general equation for a generalized Vandermonde determinant with one deleted row in terms of the elementary
symmetric polynomial \cite[p.  466]{GVD29}.

\begin{lemma}\label{lem:Vandmonde-g}
For each $\ell$ with $0\le \ell \le n$,  it holds that
\begin{eqnarray}
\begin{vmatrix}
1  & 1 & \cdots& 1 & 1\\
u_1 & u_2 & \cdots & u_{n-1} & u_n\\
\vdots & \vdots & \ddots & \vdots & \vdots\\
u_1^{\ell-1} & u_2^{\ell -1} & \cdots & u_{n-1}^{\ell-1} & u_{n}^{\ell -1}\\
u_1^{\ell+1} & u_2^{\ell +1} & \cdots & u_{n-1}^{\ell+1} & u_{n}^{\ell +1}\\
\vdots & \vdots & \ddots & \vdots & \vdots\\
u_1^{n} & u_2^{n} & \cdots & u_{n-1}^{n} & u_{n}^{n}\\
\end{vmatrix}=\left( \prod_{1\le j < i \le n} (u_i-u_j)  \right) \sigma_{n-\ell}(u_1,\dots, u_n).
\end{eqnarray}
\end{lemma}

The following lemma will be used to determine the codewords with weight $\delta+1$ in $\C_{(\delta^m,  \delta^m+1,  \delta,  1)}$.

\begin{lemma}\label{lem:min-word-subpl}
Let $\delta$ be a power of a prime $p$ and $q=\delta^m$.  Let $u_0 \in U_{q+1} \setminus \lbrace 1, -1 \rbrace$.  Set
\begin{eqnarray*}
\begin{array}{ll}
u_c=\frac{c+u_0^q}{c+u_0},  &u_{\infty} =1,\\
 a_c=(c+u_0)^{(q+1)(\delta-1)},   & a_{\infty}=1,
\end{array}
\end{eqnarray*}
where $c  \in \gf(\delta) $.  Then $u_c \in U_{q+1},  a_c \in \gf(q)^{*}$ and it holds that
\begin{eqnarray*}
\sum_{c \in \gf(\delta) \cup \lbrace \infty  \rbrace} a_c u_c^{e}=0,
\end{eqnarray*}
where $e \in \lbrace 1,2,  \dots, \delta-1 \rbrace$.
\end{lemma}
\begin{proof}
It is straightforward to check that $u_c \in U_{q+1}$ and $ a_c \in \gf(q)^{*}$.
We compute
\begin{eqnarray*}
\begin{array}{l}
\sum_{c \in \gf(\delta)} a_c u_c^{e}\\
= \sum_{c \in \gf(\delta)^*} (c+u_0^{-1})^{\delta-1+e}(c+u_0)^{\delta-1-e}+u_0^{2e} \\
=\sum_{c \in \gf(\delta)^*}  \sum_{i=0}^{\delta-1+e}  \sum_{j=0}^{\delta-1-e}  \binom{\delta-1+e}{i} \binom{\delta-1-e}{j}   u_0^{-2e+i-j}     c^{i+j}+u_0^{-2e}\\
= \sum_{i=0}^{\delta-1+e}  \sum_{j=0}^{\delta-1-e}  \binom{\delta-1+e}{i} \binom{\delta-1-e}{j}   u_0^{-2e+i-j}  \sum_{c \in \gf(\delta)^*} c^{i+j}+u_0^{-2e}.
\end{array}
\end{eqnarray*}
Together with the fact
$\sum_{c \in \gf(\delta)^*} c^{i+j}=  \left\lbrace
\begin{array}{rl}
0, & (\delta-1) \not | (i+j)\\
-1,  & (\delta-1) | (i+j)\end{array}\right.$,
we have
\begin{eqnarray*}
\begin{array}{l}
\sum_{c \in \gf(\delta)} a_c u_c^{e}\\
= -\sum_{j=0}^{\delta-1-e}  \binom{\delta-1+e}{\delta-1-j} \binom{\delta-1-e}{j}   u_0^{-2e+\delta-1-2j}+u_0^{-2e} -u_0^{-2e}-1\\
= -\sum_{j=0}^{\delta-1-e}  \binom{\delta-1+e}{\delta-1-j} \binom{r\delta-1-e}{j}   u_0^{-2e+\delta-1-2j}-1.\\
\end{array}
\end{eqnarray*}
By  Lemma \ref{eqn:q+e-1 mod},  we have
\[\sum_{c \in \gf(\delta)} a_c u_c^{e}=-1.\]
This completes the proof of the lemma.
\end{proof}

\begin{theorem}\label{thm:C-3-5}
Let $\delta$ be a power of a prime $p$ and $q=\delta^m$ with $m \ge 2$.
Then the narrow-sense antiprimitive BCH code $\C_{(q, q+1,  \delta, 1)}$  has parameters $[q+1, q-2\delta+3,  \delta+1]_q$.
\end{theorem}
\begin{proof}
Note that the generator polynomial of the narrow-sense antiprimitive BCH code $\C_{(q, q+1,  \delta, 1)}$ is
$\prod_{i=1}^{\delta-1} \left( x^2- (\beta^i+\beta^{-i})x+1 \right)$,  where $\beta$ is a primitive $(q+1)$-th root of unity.  The desired conclusion on the dimension of the code follows.

We now determine the minimum distance of the code $\C_{(q, q+1,  \delta, 1)}$.
According to the BCH bound,  the minimum distance d of the code $\C_{(q, q+1,  \delta, 1)}$ satisfies $d \ge \delta$.
In order to prove $d\ge (\delta+1)$ it suffices to prove the following claim: there is no codeword of weight $\delta$
in  $\C_{(q, q+1,  \delta, 1)}$.  Define
\begin{eqnarray}\label{eq:H}
H=\left[
\begin{array}{cccccc}
1  &  \beta^{-1\cdot (\delta-1)} & \beta^{-2\cdot (\delta-1)} & \beta^{-3\cdot (\delta-1)} & \cdots & \beta^{-q\cdot (\delta-1)}\\
\vdots  & \vdots  & \vdots &  \vdots & \cdots & \vdots \\
1  & \beta^{-1\cdot 1} & \beta^{-2\cdot 1} & \beta^{-3 \cdot 1} & \cdots & \beta^{-q \cdot 1} \\
1  & \beta^{+1\cdot 1} & \beta^{+2\cdot 1} & \beta^{+3 \cdot 1} & \cdots & \beta^{+q \cdot 1} \\
\vdots  & \vdots  & \vdots &  \vdots & \cdots & \vdots \\
1  &  \beta^{+1\cdot (\delta-1)} & \beta^{+2\cdot (\delta-1)} & \beta^{+3\cdot (\delta-1)} & \cdots & \beta^{+q\cdot (\delta-1)}
\end{array}
\right].
\end{eqnarray}
It is easily seen that $H$ is a parity-check matrix of $\C_{(q, q+1, \delta,1)}$, i.e.,
\begin{align*}\label{eq:C-H}
\C_{(q, q+1, \delta,1)}=\{\bc \in \gf(q)^{q+1}: \bc H^T=\bzero\}.
\end{align*}
Suppose that there is a codeword of weight $\delta$.  Then
there exist $\{u_1, \cdots, u_{\delta}\} \in \binom{U_{q+1}}{\delta}$ and $(x_1, \cdots, x_{\delta})\in  \left (\gf(q)^* \right )^{\delta}$
such that $M_{\delta,  \delta}(x_1, \cdots, x_{\delta})^T=0$.   Applying  Lemma \ref{eq:Mx=0},  we find that $\rank (M_{\delta, \delta}) <\delta$.
Then the determinant of the square matrix $M_{\delta,  \delta}[-(\delta-1)+i,  1+i]$ equals zero for any  $i \in \left\{ 0, 1,  \dots,  \delta -2 \right\}$.
By Lemma \ref{lem:Vandmonde-g},  we have
\[\sigma_{\ell}(u_1, \dots, u_{\delta})=0,\]
where $1 \le \ell \le \delta-1$.
By Vieta's formula,  we obtain
\[\prod_{i=1}^{\delta} (u-u_i)=u^{\delta}+(-1)^{\delta}\prod_{i=1}^{\delta} u_i,\]
where $u$ is an indeterminate.  Substituting $u$ in both sides of the above equation by $u_i$ and $u_j$ ($1 \le i < j \le \delta$),  respectively,
we find that $u_i^{\delta}=u_j^{\delta}$: a contradiction.  We thus deduce that $d \ge (\delta+1)$.

Let $u_0$ be a fixed element in  $U_{q+1} \setminus \{+ 1,  -1\}$.  Write $u(x)=\frac{x+u_0^q}{x+u_0}$ and
 $a(x)=(x+u_0)^{(q+1)(\delta-1)}$,   where $x \in \gf(\delta)$,  and $u(\infty)=1, a(\infty)=1$.
Set $\bc=(c_u)_{u \in U_{q+1}}$ where
\begin{eqnarray}
c_u=\left\{
\begin{array}{rr}
a(x)  &  \text{ if } u= u(x),\\
0,  & \text{ otherwise}, \\
\end{array}
\right.
\end{eqnarray}
where $x\in \gf(\delta) \cup \{\infty\}$.
By Lemma \ref{lem:min-word-subpl},  $\bc \in \C_{(q, q+1, \delta,1)}$ and $\wt(\bc)=\delta+1$. Thus, $d=\delta +1$.

\end{proof}

The following theorem shows that the dual $\C_{(\delta^m, \delta^m+1,  \delta, 1)}^{\perp}$  of  $\C_{(\delta^m, \delta^m+1,  \delta, 1)}$ is
an almost maximum distance separable code (almost MDS code),  where $\delta\ge 3$ is a prime power.

\begin{theorem}\label{thm:Tr(C-3-5)}
Let $\delta$ be a power of a prime $p$ and $q=\delta^m$ with $\delta\ge 3,  m \ge 2$.
Then the dual $\C_{(q, q+1,  \delta, 1)}^{\perp}$  of the narrow-sense antiprimitive BCH code $\C_{(q, q+1,  \delta, 1)}$  has parameters $[q+1, 2\delta-2,  q-2\delta+3]_q$.
\end{theorem}

\begin{proof}
By Delsarte's Theorem,   any nonzero codeword of the code $\C_{(q, q+1,  \delta, 1)}^{\perp}$  can be written as $\bc=(f(u))_{u \in U_{q+1}}$,
where $f(u)= \tr_{q^2/q}\left(\sum_{i=1}^{\delta-1} a_i u^i\right)$ and $a_i \in \gf(q^2)$.  Rewrite $f(u)$ in the form
\[f(u)=u^{-(\delta-1)} \left( a_{\delta-1} u^{2(\delta-1)} +
a_{\delta-2} u^{2\delta-3} + \dots +a_{\delta-1}^q  \right).\]
It follows that there are at most $2(\delta-1)$ values of  $u\in U_{q+1}$ such that $f(u)=0$.
That is,  the minimum distance $d^{\perp}$ of $\C_{(q, q+1,  \delta, 1)}^{\perp}$  satisfies $d^{\perp} \ge q-2\delta+3$.
On the other hand,  we have $d^{\perp} \le q-2\delta+4$ by the Singleton bound.  Suppose that $d^{\perp} = q-2\delta+4$.
Then $\C_{(q, q+1,  \delta, 1)}^{\perp}$ is a MDS code and so is $\C_{(q, q+1,  \delta, 1)}$: a contradiction to the minimum distance of
$\C_{(q, q+1,  \delta, 1)}$ given in Theorem \ref{thm:C-3-5}.
This completes the proof.
\end{proof}

Let $\gf(q)^{U_{q+1}}$ denote the vector space consisting  of all elements $(c_u)_{u \in U_{q+1}}$,
where $c_u\in \gf(q)$.  The action of the  semidirect product $\left( \gf(q)^* \right)^{U_{q+1}} \rtimes \mathrm{Stab}_{U_{q+1}}$
on $\gf(q)^{U_{q+1}}$
is defined by
\begin{eqnarray*}
\left((a_u)_{u\in U_{q+1}};  g\right) (c_u)_{u\in U_{q+1}}  = (a_{u} c_{g^{-1}u})_{u \in U_{q+1}}.
\end{eqnarray*}
Thus the multiplication in $\left( \gf(q)^* \right)^{U_{q+1}} \rtimes \mathrm{Stab}_{U_{q+1}}$ is given by
\begin{eqnarray*}
\left((a_u)_{u\in U_{q+1}};  g_1\right) \left((b_u)_{u\in U_{q+1}};  g_2\right) = \left((c_u)_{u\in U_{q+1}};  g_1g_2\right),
\end{eqnarray*}
where $c_u=a_u b_{g_1^{-1}u}$.

Note that $\C_{(q, q+1, \delta, 1)}^{\perp}=\left\{ (f(u))_{u \in U_{q+1}}:  f \in \cP(\delta, q) \right\}$. Then the results on the group representation in Lemma \ref{lem:representation} translate immediately into corresponding results
on the monomial isomorphisms of $ \C_{(q, q+1, \delta, 1)}^{\perp}$.
\begin{theorem}\label{thm:C-dual-group}
Let $\delta$ be a power of a prime $p$ and $q=\delta^m$ with $m \ge 2$.
Define a subgroup of $\left( \gf(q)^* \right)^{U_{q+1}} \rtimes \mathrm{Stab}_{U_{q+1}}$ by
\begin{eqnarray*}
G_{\delta}^{\perp}=\left\{ \left(\left((cu+d)^{(q+1)(\delta-1)}\right)_{u\in U_{q+1}}; \left( \begin{array}{cc}
d^q & c^q\\
c & d
\end{array} \right)^{-1} \right): c,  d \in \gf(q^2),  c^{q+1} \neq d^{q+1}\right\}.
\end{eqnarray*}
Then $G_{\delta}^{\perp}$ is a subgroup of the monomial automorphism group $\mathrm{MAut}(\C_{(q, q+1, \delta, 1)}^{\perp})$.
In particular,  the automorphism group of $\C_{(q, q+1, \delta, 1)}^{\perp}$ is $3$-transitive.
\end{theorem}

The following theorem is  an immediate consequence of Theorem \ref{thm:C-dual-group}.
\begin{theorem}\label{thm:C-group}
Let $\delta$ be a power of a prime $p$ and $q=\delta^m$ with $m \ge 2$.
Define a subgroup of $\left( \gf(q)^* \right)^{U_{q+1}} \rtimes \mathrm{Stab}_{U_{q+1}}$ by
\begin{eqnarray*}
G_{\delta}=\left\{ \left(\left((cu+d)^{-(q+1)(\delta-1)}\right)_{u\in U_{q+1}}; \left( \begin{array}{cc}
d^q & c^q\\
c & d
\end{array} \right)^{-1} \right): c,  d \in \gf(q^2),  c^{q+1} \neq d^{q+1}\right\}.
\end{eqnarray*}
Then $G_{\delta}$ is a subgroup of the monomial automorphism group $\mathrm{MAut}(\C_{(q, q+1, \delta, 1)})$.
In particular,  the automorphism group of $\C_{(q, q+1, \delta, 1)}$ is $3$-transitive.
\end{theorem}

The following theorem presents a coding-theoretic construction of  the Witt spherical geometry designs.
This theorem shows that the supports of the codewords of minimum Hamming weight in the BCH code
$\C_{(q, q+1, \delta, 1)}$
 yield a Witt spherical geometry design.

\begin{theorem}\label{thm:coding-designs}
Let $\delta$ be a power of a prime $p$ and $q=\delta^m$ with $m \ge 2$.
Then the incidence structure $\left ( U_{q+1},  \mathcal B_{\delta+1} \left (\C_{(q, q+1, \delta, 1)}\right ) \right)$
is isomorphic to the
Witt spherical geometry design with parameters $3$-$(\delta^m+1,\delta+1,1)$.
\end{theorem}

\begin{proof}
By Theorems \ref{thm:C-3-5} and  \ref{thm:C-group},  $\left ( U_{q+1},  \mathcal B_{\delta+1} \left (\C_{(q, q+1, \delta, 1)}\right ) \right)$
is a $3$-$(\delta^m+1, \delta+1, \lambda)$,  where $\lambda$ is a positive integer.
Let $u_0$ be a fixed generator element of  $U_{q+1}$.  Then we have
$\lambda= \left| \cB_{u_0}  \right|$,
where $\cB_{u_0}=\left\lbrace  B \in \mathcal B_{\delta+1} \left (\C_{(q, q+1, \delta, 1)}\right ) :  \{1, u_0, u_0^2\}\subseteq B\right\rbrace$.
Let $ B =\left\{ u_1,  \dots,  u_{\delta+1}\right\} \in \cB_{u_0}$ with $u_1=1,  u_2= u_0,  u_3= u_0^2$.  Recall that
\begin{align*}
\C_{(q, q+1, \delta,1)}=\{\bc \in \gf(q)^{q+1}: \bc H^T=\bzero\},
\end{align*}
where $H$ is given by (\ref{eq:H}).
By Lemma \ref{lem:sol-rank},  the rank of the matrix $M_{\delta,  \delta+1}$ defined in (\ref{eq:M}) is less than $\delta+1$.
So the determinant of the square matrix $M_{\delta,  \delta+1}[-(\delta-1)+i,  2+i]$ equals zero,  where $0 \le i \le \delta-3$.
By Lemma \ref{lem:Vandmonde-g},  we have
\begin{eqnarray*}
\sigma_{\ell}(u_1,  \dots, u_{\delta+1})=0,
\end{eqnarray*}
where $2\le \ell \le \delta -1$.
By Vieta's formula,  we obtain
\[\prod_{i=1}^{\delta+1} (u-u_i)=u^{\delta+1}+a u^{\delta} +b u+c,\]
where $u$ is an indeterminate and $(a,b,c)\in \gf(q^2)^3$.
Substituting $u$ in both sides of the above equation by $u_i$ ($1 \le i  \le 3$),  we get
\begin{eqnarray*}
\left\{ \begin{array}{rrrl}
a+&  b+ & c& =-1\\
u_0^{\delta}a+&  u_0b+ & c& =-u_0^{\delta+1}\\
u_0^{2\delta}a+&  u_0^2b+ & c& =-u_0^{2(\delta+1)}\\
\end{array} \right. .
\end{eqnarray*}
Note that the coefficient matrix of the above system of  equations is nonsingular.
Thus the system  has a unique solution of $(a,  b, c)$.
That is,  $(a,  b, c)$ is uniquely determined by $u_0$.
In particular,  $B$ must be equal to $\left \{u \in U_{q+1}:  u^{\delta+1}+a u^{\delta} +b u+c=0\right\}$.
To conclude,  $\lambda$ is equal to $1$.

Let $g_0=\left( \begin{array}{cc}
u_0 & 1\\
1 & u_0
\end{array} \right)$.  By Lemma \ref{lem:min-word-subpl},  $\frac{1}{u_0} g_0 \PG(1, \delta) \in \mathcal B_{\delta+1} \left (\C_{(q, q+1, \delta, 1)}\right )$.
It follows that $g_0 \PG(1, \delta) \in \mathcal B_{\delta+1} \left (\C_{(q, q+1, \delta, 1)}\right )$ from Theorem \ref{thm:C-group}.
Since $\lambda=1$,  by (\ref{eq:stabilizers}),  we conclude that
\begin{eqnarray*}
\mathcal B_{\delta+1} \left (\C_{(q, q+1, \delta, 1)}\right )&=&\mathrm{Stab}_{U_{q+1}}\left(  g_0 \PG(1, \delta) \right) \\
&=&\left \{ g_0g  \PG(1, \delta):  g \in \PGL(2,q)\right\}\\
&=& g_0 \mathrm{Orb}_{\PG(1, \delta)} ,
\end{eqnarray*}
where $\mathrm{Orb}_{\PG(1, \delta)}= \left \{ g  \PG(1, \delta):  g \in \PGL(2,q)\right\}$.
Therefore,  the map from $\PGL(2,q)$ to $U_{q+1}$ given by $x \mapsto g_0 x$ gives rise to
an isomorphism  between the two incidence structures $\left(\PGL(2,q),  \mathrm{Orb}_{\PG(1, \delta)} \right)$ and $\left ( U_{q+1},  \mathcal B_{\delta+1} \left (\C_{(q, q+1, \delta, 1)}\right ) \right)$.  The desired conclusion then follows from the definition of the
Witt spherical geometry design.
\end{proof}

Combining Theorems  \ref{thm:C-group} and  \ref{thm:coding-designs} with Theorem  \ref{thm;2-rank} yields the following result.
\begin{theorem}\label{thm:Witt-rank}.
Let $\delta$ be a power of a prime $p$ and $m$ an integer with $m \ge 2$.
Then the $p$-rank of   the
Witt spherical geometry design with parameters $S(3,  \delta+1,  \delta^m+1)$
is $\delta^m+1$.
\end{theorem}

\begin{example}
Let $q=25$.  Then the narrow-sense BCH code $\C_{(q,  q+1,  5,  1 )}$ has parameters $[26,18,6]_{25}$
and weight enumerator
\begin{eqnarray*}
1+3120z^6+1053000 z^8\\
+52478400 z^9+ 2246164440 z^{10}\\
+ 76730209920 z^{11} + 2313008100000 z^{12} \\
+ 59737548888000 z^{13}+ 1331420089708800 z^{14}\\
 + 25563001945153920 z^{15} + 421789956437369520 z^{16}\\
 + 5954681202248610000  z^{17}+ 71456174963080050000 z^{18} \\
 + 722083451831987107200 z^{19} +6065500995657406236960 z^{20}\\
  + 41592006827232472278720 z^{21}+ 226865491784954611290000 z^{22}\\
  + 946916835276318186384000 z^{23}+2840750505828957328830600 z^{24} \\
  + 5454240971191597731710304 z^{25} + 5034683973407628695013720 z^{26},
\end{eqnarray*}
and
$\left ( U_{q+1},  \mathcal B_{6} \left (\C_{(q, q+1, \delta, 1)}\right ) \right)$
is a $3$-$(26,  6,  1)$ design.

The dual code  $ \C_{(q,  q+1,  5,  1 )}^{\perp}$ has parameters $[26,8,18]_{25}$
and weight enumerator
\begin{eqnarray*}
1+ 1645800 z^{18}+ 4180800 z^{19}+ 70265520 z^{20}+ 426192000 z^{21}\\
+ 2393352000 z^{22}+ 9911491200 z^{23} + 29801335200 z^{24}\\
+ 57185869104 z^{25}+ 52793559000 z^{26},
\end{eqnarray*}
and
$\left ( U_{q+1},  \mathcal B_{18} \left (\C_{(q, q+1, \delta, 1)}^{\perp}\right ) \right)$ is a $3$-$(26,  18,21522)$ design.
\end{example}

\section{Summary and concluding remarks}\label{sec:concl}

We have provided a detailed discussion of the interplay among the narrow-sense antiprimitive BCH codes,
group actions and group representations concerning $\PGL(2,q)$,  and combinatorial $3$-designs in this paper.
The main contributions of this paper are the following:
\begin{itemize}
\item Infinite families of narrow-sense antiprimitive BCH codes admitting a $3$-transitive automorphism group were
proposed in Theorem \ref{thm:C-group}.  The dimensions and the minimum distances of these codes and their duals were also determined
in Theorems \ref{thm:C-3-5} and \ref{thm:Tr(C-3-5)}.  Using Delsarte's Theorem, it is shown in Theorem \ref{thm:Tr(C-3-5)} that the dual codes of the narrow-sense antiprimitive BCH codes derived in this paper are almost MDS.

\item A coding-theory construction of the Witt spherical geometry design $S(3,  \delta+1,  \delta^m+1)$ was presented in Theorem \ref{thm:coding-designs}.

\item A complete classification of $\PGL(2,  p^m)$-invariant $p^h$-ary linear codes was
established 
in Theorem \ref{thm:code-PGL-4}.

\item The $p$-ranks of incidence structures that are invariant under the action of $\PGL(2, p^m)$ were derived
in Theorem \ref{thm;2-rank}.  In particular,  the $p$-rank of the Witt spherical geometry design $S(3,  \delta+1,  \delta^m+1)$
was determined in Theorem \ref{thm:Witt-rank}.
\end{itemize}
The results of this paper generalize and extend the work in \cite{DTT21}.
It would be interesting to determine structures and parameters for more $3$-designs held in
$\C_{(\delta^m, \delta^m+1, \delta, 1)} $ and $ \C_{(\delta^m, \delta+1, \delta, 1)}^{\perp}$,
where $\delta$ is a prime power.  It would be valuable to determine the full automorphism groups of the BCH codes introduced in this paper.

Finally, we would explain an important motivation for constructing a linear code over a finite field to support a known $t$-design constructed with an algebraic, combinatoric, or group-theoretic approach. Such an investigation may not be interesting in combinatorics, as the known $t$-design was already discovered earlier. However, this would be very interesting in coding theory and enhances coding theory, as the newly discovered linear code supporting the known $t$-design should have special properties compared with general linear codes.
It is known that the dual code of such a code admits majority-logic decoding \citep{RB75,Rud67,Tonch98}.
By definition, $t$-designs have a certain level of symmetry. The larger the strength $t$ of a $t$-design is, the higher the level of symmetry the $t$-design has. For instance, all the known linear codes supporting a
$4$-design or a $5$-design have special properties \cite{Dingbook18,TD20}.
While the Witt spherical geometry design $S(3, \delta+1, \delta^m+1)$ was discovered 80 years ago, the effort of constructing a linear code supporting this design carried out in this paper has led to the discovery of the codes $\C_{(\delta^m, \delta^m+1, \delta, 1)}$
which are $\PGL(2, \delta^m)$-invariant and a complete classification of $\PGL(2, p^m)$-invariant $p^h$-ary linear codes.
In addition, an infinite family of almost MDS codes $\C_{(\delta^m, \delta^m+1, \delta, 1)}^\perp$ were obtained in this project.
The results obtained in this paper enhance coding theory and strengthen the interplay between coding theory and design theory.
In addition to the Witt spherical geometry design $S(3, \delta+1, \delta^m+1)$, more $3$-designs with new parameters are supported by the codewords of other weights in $\C_{(\delta^m, \delta^m+1, \delta, 1)}$. Hence,
this paper does have contributions to the theory of combinatorial designs.
Furthermore, the codes presented in this paper also have applications in cryptography (secret sharing \cite{YD06} and authentication codes \cite{DHKW}).





\end{document}